\documentclass[10pt, draft]{article}

\usepackage{jucs2e}
\usepackage[cmex10]{amsmath}
\usepackage{amssymb}
\usepackage{xspace}
\usepackage{prooftree}
\usepackage{cite}
\usepackage{calc}

\newcommand{\textsubscript}[1]{\ensuremath{_{\textrm{#1}}}}
\makeatletter
\newcommand*{\sourceatright}[2][2em]{{
  \unskip\nobreak\hfil\penalty50
  \hskip#1\hbox{}\nobreak\hfil{\em #2}
  \parfillskip\z@\finalhyphendemerits=0\par}
}
\makeatother

\let\imp=\rightarrow
\newcommand{\CC}{CC\xspace}
\newcommand{\CCr}{CC\textsubscript{r}\xspace}

\DeclareMathOperator{\VL}{\mathcal{V}}                   
\DeclareMathOperator{\Type}{Type}               
\DeclareMathOperator{\Prop}{Prop}               

\newcommand{\drule}[3][ ]{
\begin{prooftree}
#2
\using{\text{#1}}
\justifies
#3
\end{prooftree}
}
\newcommand{\linespace}{\hspace*{3em}}
\newcommand{\dsub}[2]{[#1\leftarrow #2\,]}        
\newcommand{\dlineone}[1]{$$ #1 $$}             
\newcommand{\dlinetwo}[2]{$$ #1 \linespace #2 $$} 

\newcommand{\beq}{=_\beta}
\newcommand{\bred}{\rightsquigarrow_{\beta}}
\newcommand{\breds}{\stackrel{*}{\rightsquigarrow}_{\beta}}

\newcommand{\dseq}[2][\Gamma]{#1 \vdash #2}     
\newcommand{\wf}[1]{#1\ensuremath{\text{\em\xspace wf}}}
\newcommand{\dseqr}[2][\Gamma]{#1 \vdash_{\hspace{-.3em}r} #2}  

\newcommand{\wfr}[1]{#1\ensuremath{\text{\em\xspace wf\textsubscript{r}}}}
\newcommand{\dseqp}[2][\Gamma]{#1 \vdash_{\hspace{-.3em}p} #2}  

\newcommand{\wfp}[1]{#1\ensuremath{\text{\em\xspace wf\textsubscript{p}}}}

\newcommand{\norm}[1]{\|#1\|}
\newcommand{\cpl}[2]{\langle #1,#2\rangle}

\begin{document}

\title{Investigations on a Pedagogical Calculus of Constructions}
\author{
   {\bfseries Lo\"ic Colson}\\
   (LITA, University Paul-Verlaine -- Metz, France\\
   colson@univ-metz.fr)
   \and
   {\bfseries Vincent Demange}\\
   (LITA, University Paul-Verlaine -- Metz, France\\
   demange@univ-metz.fr)
}
\maketitle
\begin{abstract}
\boldmath
In the last few years appeared \textit{pedagogical propositional natural deduction systems}. In these systems, one must satisfy the \textit{pedagogical constraint}: the user must give an {\em example} of any introduced notion. In formal terms, for instance in the propositional case, the main modification is that we replace the usual rule (hyp) by the rule (p-hyp)
$$
\begin{array}{c@{\hspace{3em}}c}
\begin{prooftree}
F \in \Gamma
\using{\text{(hyp)}}
\justifies
\Gamma \vdash F
\end{prooftree}

&

\begin{prooftree}
F \in \Gamma \quad \vdash \sigma \cdot \Gamma
\using{\text{(p-hyp)}}
\justifies
\Gamma \vdash F
\end{prooftree}
\end{array}
$$
where $\sigma$ denotes a substitution which replaces variables of $\Gamma$ with an example. This substitution $\sigma$ is called the {\em motivation} of $\Gamma$.

First we expose the reasons of such a constraint and properties of these ``pedagogical'' calculi: the absence of negation at logical side, and the ``usefulness'' feature of terms at computational side (through the Curry-Howard correspondence). Then we construct a simple pedagogical restriction of the calculus of constructions (\CC) called \CCr. We establish logical limitations of this system, and compare its computational expressiveness to G\"odel system T.

Finally, guided by the logical limitations of \CCr, we propose a formal and general definition of what a pedagogical calculus of constructions should be.
\end{abstract}


\begin{category}
F.1.1, F.4.1
\end{category}

\section{Introduction and Motivations}\label{sec_intro}

\subsection{The pedagogical constraint}

Recently the articles \cite{LCDM07,LCDM08,LCDM09} appeared in print, introducing {\em pedagogical natural deduction systems} and {\em pedagogical typed $\lambda$-calculi}. The main feature about these systems is that any proof (or any program) must satisfy the so named {\em pedagogical constraint}: in natural deduction systems (for instance) the rule (hyp) is replaced by (p-hyp)

$$
\begin{array}{c@{\hspace{3em}}c}
\begin{prooftree}
F \in \Gamma
\using{\text{(hyp)}}
\justifies
\Gamma \vdash F
\end{prooftree}

&

\begin{prooftree}
F \in \Gamma \quad \vdash \sigma \cdot \Gamma
\using{\text{(p-hyp)}}
\justifies
\Gamma \vdash F
\end{prooftree}
\end{array}
$$

where $\sigma$ denotes a substitution which replaces propositional variables of $\Gamma$ with an example, and $\vdash \sigma \cdot \Gamma$ stands for the derivations of those substituted formulas.

The idea of such a constraint is that, in order to assume a set $\Gamma$ of hypotheses, one must first provide a ``motivation'' (the substitution $\sigma$ under consideration) in which the set of hypotheses is fulfilled. In doing so, we can always exemplify introduced hypotheses. This is the formal counterpart of the usual informal teaching practice, consisting in giving examples of objects satisfying the assumed properties. This last point is a justification of the terminology {\em pedagogical systems}, and the necessity of such a constraint was already observed by \cite{HP13} [see Section~\ref{sec_poincare}].

\subsection{The pedagogical minimal propositional calculus}

In~\cite{LCDM07}, the minimal propositional calculus over $\imp$, $\vee$ and $\wedge$ has been constrained as previously explained. It is shown in the article that the resulting calculus (P-MPC) is equivalent to the original one: a judgment $\Gamma \vdash F$ is derivable in the usual system (MPC) if and only if it is derivable in its pedagogical version (P-MPC).

\subsection{The pedagogical second-order propositional calculi}
The case of the second-order propositional calculus (Prop$^2$) is considered in~\cite{LCDM08}. Constraining only the rule of hypothesis as above, one is led to a {\em weakly pedagogical second-order calculus} (P$_s$-Prop$^2$), where rules dealing with quantification are the usual ones:
$$
\begin{array}{c@{\hspace{3em}}c}
\begin{prooftree}
\Gamma \vdash F \quad \alpha \not\in\VL(F)
\using{(\forall_i)}
\justifies
\Gamma \vdash \forall \alpha.F
\end{prooftree}

&

\begin{prooftree}
\Gamma \vdash \forall \alpha.F
\using{(\forall_e)}
\justifies
\Gamma \vdash F[\alpha \leftarrow U]
\end{prooftree}
\end{array}
$$

The same remark as above holds for this calculus, but it is {\em not stable by normalization} of proofs. Indeed, it is shown that $\bot \imp \bot$ is derivable in P$_s$-Prop$^2$ (where $\bot$ stands for $\forall \alpha. \alpha$):
$$
\begin{array}{l@{\,}l@{~}l}
1.& \beta \vdash \beta                       & (\text{$\beta$ is motivable}) \\
2.& \vdash \beta \imp \beta                  & (\imp_i\ 1) \\
3.& \vdash \forall \beta. \beta \imp \beta   & (\forall_i\ 2) \\
4.& \vdash \bot \imp \bot                    & (\forall_e\ 3)
\end{array}
$$

But a normal form of this proof must end with a ($\imp_i$) rule of $\bot$, which is impossible since $\bot$ is not motivable. Hence the normal form of this proof is not a proof of P$_s$-Prop$^2$.

This motivates the more constrained system P-Prop$^2$ where the ($\forall_e$) rule has been replaced by
$$
\begin{prooftree}
\Gamma \vdash \forall \alpha.F  \quad  \vdash \sigma \cdot U
\using{(\text{P-}\forall_e)}
\justifies
\Gamma \vdash F[\alpha \leftarrow U]
\end{prooftree}
$$

It is shown about this system that the usual second-order encoding of connectives $\vee$ and $\wedge$ essentially works but it must be observed that the $\vee_i$ (at right for instance) becomes:
$$
\begin{prooftree}
\Gamma \vdash A  \quad  \vdash \sigma \cdot B
\using{(\vee_{ir})}
\justifies
\Gamma \vdash A \vee B
\end{prooftree}
$$

The main result concerning P-Prop$^2$ is that there exists a translation $F \mapsto F^\gamma$ inspired by the A-translation of \cite{HF78} such that: $\Gamma \vdash F$ is derivable in Prop$^2$ if and only if $\Gamma^\gamma \vdash F^\gamma$ is derivable in P-Prop$^2$.

\subsection[The pedagogical second-order lambda-calculus]{The pedagogical second-order $\lambda$-calculus}

Through the Curry-Howard isomorphism, previous work about second-order propositional calculus is extended in~\cite{LCDM09} to the second-order $\lambda$-calculus. The system is shown to be stable by reduction (i.e. enjoys the so-called subject reduction property). An important feature for a $\lambda$-calculus is defined: the {\em usefulness} of functions. It means that every typable function in this pedagogical $\lambda$-calculus can be applied to a term: if $\vdash f: A \imp B$, then there is a substitution $\sigma$ such that $\sigma \cdot A$ is inhabited.
Indeed, pedagogical $\lambda$-calculi do not allow one to write useless programs, which are not needed.

\subsection{The calculus of constructions}

The calculus of constructions (\CC) has been first introduced in~\cite{CH84,TC85}: it is a $\lambda$-calculus which encompasses higher-order $\lambda$-calculi and calculi with dependent types. It is then natural to extend previous works on ``pedagogization'' to \CC in the aim of obtaining a uniform treatment of pedagogical $\lambda$-calculi.

\subsection{Organization of the article}
The paper is organized as follows: in section~\ref{sec_background} we recall usual notations for the calculus of constructions (\CC); in section~\ref{sec_pedagogizing} we introduce the main criterion for a subsystem of \CC to be pedagogical, we discuss about the impossibility of a straightforward modification of \CC, and we propose a better one; then in section~\ref{sec_ccr_pedag} we show that this restriction meets this criterion; we present some limitations of it at logical and computational side in sections~\ref{sec_limitations} and~\ref{sec_expressivness}; finally we conclude by the first formal definition of a pedagogical subsystem of \CC.
\section{Background and Notations}\label{sec_background}

In this section, we briefly recall usual definitions and notations about the calculus of constructions \CC.

We try to use $x,y,..$ as symbols for variables, $u,v,w,t,..$ to denote terms, and $A,B,..$ for types and formulas.

$\equiv$ is the syntactical equality of terms\,\footnote{As in \cite{TC89}, we assume De Bruijn indexes for bound variables and identifiers for free variables. So there is no need for $\alpha$-conversion notion.}. We note by $\bred$ the usual beta-reduction relation between terms; $\breds$ its reflexive and transitive closure; and $\beq$ its equivalence closure. $\VL(t)$ is the set of free variables of $t$. $t$ is said to be closed if $\VL(t) = \emptyset$. $t\dsub{x}{u}$ is the usual substitution of $u$ for $x$ in $t$; and $t\dsub{x_1,..,x_n}{u_1,..,u_n}$ is the simultaneous substitution of $u_1$ for $x_1$, $u_2$ for $x_2$, etc. To shorten notations, we use a vector symbolism: $\vec{t}$ denotes the sequence of terms $t_1,..,t_n$; and $\forall \vec{x}^{\vec{A}}.B$ denotes $\forall x_1^{A_1}..\forall x_n^{A_n}.B$.

There are two kinds of judgments: $\wf{\Gamma}$ means that the environment $\Gamma$ is syntactically well-formed, and $\dseq{t:A}$ expresses that the term $t$ is of type $A$ in the environment $\Gamma$. Implicitly $\dseq{A:\kappa}$ signifies that there exists $\kappa \in \{\Prop,\Type\}$ such that this previous statement holds. $\dseq{t:A:\kappa}$ is the contraction of $\dseq{t:A}$ and $\dseq{A:\kappa}$. As usual, $A \imp B$ is a shortcut notation for $\forall x^A.B$ when $x$ does not appear in $B$.

Rules of \CC are presented in [Fig.~\ref{fig_cc_rules}]: a close presentation can be found in~\cite{MBJPS04} (without the well-formed judgment), or in~\cite{TC86,HB92}.

\begin{figure*}
\fbox{\begin{minipage}{\textwidth-.95em}

\begin{tabular}{r@{\hspace*{2em}}l}

\drule[(env\textsubscript{1})]{}{\wf{[\,]}}
&
\drule[(env\textsubscript{2})]{
\dseq{A:\kappa} \quad x \not\in \VL(\Gamma)
}{
\wf{\Gamma, x:A}
}
\\[3em]

\drule[(ax)]{\wf{\Gamma}}{\dseq{\Prop:\Type}}
&
\drule[(var)]{
\wf{\Gamma,x:A,\Gamma'}
}{
\dseq[\Gamma,x:A,\Gamma']{x:A}
}
\\[3em]

\drule[(abs)]{
\dseq[\Gamma,x:A]{u:B:\kappa}
}{
\dseq{\lambda x^A.u : \forall x^A.B}
}
&
\drule[(prod)]{
\dseq[\Gamma,x:A]{B:\kappa}
}{
\dseq{\forall x^A.B:\kappa}
}
\\[3em]

\drule[(app)]{
\dseq{u:\forall x^A.B} \quad \dseq{v:A}
}{
\dseq{u\ v:B\dsub{x}{v}}
}
&
\drule[(conv)]{
\dseq{t:A} \quad \dseq{A':\kappa} \quad A \beq A'
}{
\dseq{t:A'}
}
\end{tabular}

\vspace*{1.5em}

\centering where $\kappa$ stands for $\Prop$ or for $\Type$.
\end{minipage}}
\caption{Inference rules of \CC}
\label{fig_cc_rules}
\end{figure*}

\bigskip

Beta-reduction is known to be confluent and terms of this calculus to be strongly normalizing \cite{HB92}.

\noindent In the sequel we shall need the following elementary results (proofs in \cite{TC85,HB92}):
\begin{lemma}\label{lem_type_nowhere}
If $\wf{\Gamma}$ holds, then $\Type \not\in \Gamma$ (the constant $\Type$ never appears in any well-formed environment). And if $\dseq{t:A}$ holds, then $\Type \not\in \Gamma \cup \{t\}$.
\end{lemma}

\begin{lemma}\label{lem_type_of_type}
If $\dseq{t:A}$ holds, then $A \equiv \Type$ or $\dseq{A:\kappa}$.
\end{lemma}

\begin{proposition}\label{prop_subst}
\begin{itemize}
	\item[ (i)] If $\wf{\Gamma, x:A, \Gamma'}$ and $\dseq{u:A}$ hold, then $\wf{\Gamma,\Gamma'\dsub{x}{u}}$ also holds.
	\item[(ii)] If $\dseq[\Gamma, x:A, \Gamma']{t:B}$ and $\dseq{u:A}$ hold, then $\dseq[\Gamma,\Gamma'\dsub{x}{u}]{t\dsub{x}{u}:B\dsub{x}{u}}$ holds.
\end{itemize}
\end{proposition}
\section{Pedagogizing \CC}\label{sec_pedagogizing}

\subsection{The Poincar\'e criterion}\label{sec_poincare}

Let us recall the necessity of the pedagogical constraint ---here in the case of definitions by postulate--- by the following quotation:
\begin{quotation} \em
A definition by postulate has value only when the existence of the object defined has been proved. In mathematical language, this means that the postulate does not imply a contradiction, we do not have the right to neglect this condition. Either it is necessary to admit the absence of contradiction as an intuitive truth, as an axiom, by a kind of act of faith ---but then it is necessary to realize what we are doing and to remember that we have extended the list of indemonstrable axioms--- or else it is necessary to construct a formal proof, either {\em by means of examples} or by the use of reasoning by recurrence. Not that this proof is less necessary when a direct definition is involved, but it is generally easier.
\sourceatright{Henri Poincar\'e -- Last thoughts~\cite{HP13}}
\end{quotation}

In \CC, a definition by postulate of an object $x$ may be seen as an environment containing $x$ followed by hypotheses about $x$. For instance,
\begin{quote}
Let $x$ be a natural number verifying $P(x)$ and $Q(x)$.
\end{quote}
is formally represented in \CC by the following environment
$$
x : \mathbb{N}, H_1: P(x), H_2: Q(x)
$$

Poincar\'e pointed out that such a set of hypotheses is an admissible definition by postulate of $x$ {\em only if} we are able to exhibit a natural satisfying both predicates $P$ and $Q$. In other words, types $P(x)$ and $Q(x)$ must be inhabited for a given $x$ (say $n$) in \CC. Namely the following statements must hold:
$$
\begin{array}{c@{\qquad}c@{\qquad}c}
\dseq[]{n:\mathbb{N}} &
\dseq[]{t_1:P(n)} &
\dseq[]{t_2:Q(n)}
\end{array}
$$
If this is not possible (i.e. there is no such $n$, $t_1$ or $t_2$) then the definition is meaningless and should be avoided.

Let us generalize to any environment:

\begin{definition}[(Poincar\'e criterion)]
The environment $x_1:A_1,..,x_n:A_n$ is respectful of the Poincar\'e criterion {\em only if} there exists terms $t_1,..,t_n$ such that the following judgments are derivable:
$$
\begin{array}{c}
\dseq[]{t_1:A_1} \\
\dseq[]{t_2:A_2\dsub{x_1}{t_1}} \\
\vdots \\
\dseq[]{t_n:A_n\dsub{x_1,\ldots,x_{n-1}}{t_1,\ldots,t_{n-1}}}
\end{array}
$$

A formal system is said to meet the Poincar\'e criterion {\em only if} every well-formed environment are respectful of the Poincar\'e criterion.
\end{definition}

\subsection{On the naive extension of previous work}\label{sec_naive_extension}
In the previous works on pedagogization [see section~\ref{sec_intro}], each environment is motivated before being used. It is then immediate that each used environment can be motivated, hence such a system trivially satisfies the Poincar\'e criterion. Unfortunately such a simple adjustment can not be performed into \CC.

\medskip

The straightforward extension of the previous work to \CC can be summed up by the following changes:
\begin{itemize}
\item remove (env\textsubscript{1}) and (env\textsubscript{2}) rules;
\item replace (ax) and (var) rules by these ones:
	\dlinetwo{
		\drule[(ax)]{\sigma\cdot\Gamma}{\dseq{o:\top:\Prop:\Type}}
	}{
		\drule[(var)]{
		\sigma\cdot(\Gamma,x:A,\Gamma')
		}{
		\dseq[\Gamma,x:A,\Gamma']{x:A}
		}
	}
\end{itemize}
where
\begin{itemize}
\item $\sigma$ is the substitution $[x_1 \mapsto t_1; \ldots; x_n \mapsto t_n]$ when $\Gamma\equiv x_1:A_1, \ldots, x_n:A_n$, and $\sigma\cdot\Gamma$ denotes the judgments:
$$
\begin{array}{c}
\dseq[]{t_1:A_1} \\
\dseq[]{t_2:A_2\dsub{x_1}{t_1}} \\
\vdots \\
\dseq[]{t_n:A_n\dsub{x_1,\ldots,x_{n-1}}{t_1,\ldots,t_{n-1}}}
\end{array}
$$
\item $o$ and $\top$ are two added constants in order to be able to begin derivations (like in\cite{LCDM09}).
\end{itemize}
In this subsection, we refer to this system as $P$, and index its judgments by $p$.

\bigskip

$P$ is not a subsystem of \CC:
\begin{lemma}
The following derivations hold in $P$ but not in \CC:
\begin{enumerate}
 \item[(a)] $\dseqp[x_1:\Type]{\Prop:\Type}$
 \item[(b)] $\dseqp[x_1:\Prop, x_2:(\lambda H^{\top \imp x_1}.\top)\ (\lambda y^{\top}.y) ]{\Prop:\Type}$
 \item[(c)] $\dseqp[x_1:\mathbb{N}, x_2:(\lambda H^{x_1=0}.\top)\ (\lambda P^{\mathbb{N} \imp \Prop}.\lambda H^{P\ 0}.H) ]{\Prop:\Type}$
\end{enumerate}
\end{lemma}
\begin{proof}
Proofs that derivations hold in $P$ are trivial as soon as we exhibit a motivation:
\begin{enumerate}
 \item[(a)] $\sigma_1:=[x_1 \mapsto \Prop]$
 \item[(b)] $\sigma_2:=[x_1 \mapsto \top; x_2 \mapsto o]$
 \item[(c)] $\sigma_3:=[x_1 \mapsto 0; x_2 \mapsto o]$
\end{enumerate}
And it is easy to see that they are not derivable in \CC:
\begin{enumerate}
 \item[(a)] $\Type$ appears into an environment, which is forbidden in \CC [see lemma~\ref{lem_type_nowhere}];
 \item[(b)] $(\lambda H^{\top \imp x_1}.\top)\ (\lambda y^{\top}.y)$ is ill-typed since the function waits for a element of type $\top \imp x_1$, but an element of type $\top \imp \top$ is given instead;
 \item[(c)] same reason as for (b): the function waits for a proof of $x_1=0$, whereas a proof of $0 = 0$ is passed.
\end{enumerate}
\qed
\end{proof}

\begin{remark}
Those examples involve dependent types. It seems that this naive extension can work for $\lambda^{\omega}$ [see\cite{DM08}].
\end{remark}
\begin{remark}
The first case can be avoided by enforcing the $A_i$ to be of type $\Prop$ or $\Type$ in the definition of $\sigma\cdot\Gamma$.
\end{remark}

\CC has the advantage that well-formed types are built into the system. So we just need to find which rules need to be constrained and how in order to avoid not motivable types (i.e. empty types).

\subsection{A simple attempt: \CCr}

In \CC, we are able to introduce $\bot:=\forall A^{\Prop}.A$ as an hypothesis if we have been able to derive $\bot$ as a type, which is allowed by the (prod) rule. Actually, the (prod) rule is the only one able to create vacuity, since other rules construct types and an inhabitant of it simultaneously. We then impose products to always be inhabited by replacing the usual (prod) rule of \CC by the following more restrictive one:
$$
	\drule[(prod\textsubscript{r})]{
	\dseqr[\Gamma,x:A]{\mathbf{t}:B:\kappa}
	}{
	\dseqr{\forall x^A.B:\kappa}
	}
$$
This rule may be condensed together with (abs) to obtain a rule with two conclusions. So the resulting calculus can be viewed as \CC without the (prod) rule.

From now on we will refer to the resulting calculus as \CCr, whose judgments will be indexed by $r$.

Usual properties of \CC from \cite{TC85} still hold for this calculus, especially substitution (prop.\ref{prop_subst} above), weakening and the well-known ``subject reduction'' (stability by reduction). These were formally checked in the Coq proof assistant by straightforward adaptation of the work in\cite{Bar96}.

\subsection*{Example of derivation in \CCr}
\begin{lemma}\label{cocr_not_empty}
The following rule is derivable:
$$
\begin{prooftree}
\wfr{\Gamma}
\justifies
\dseqr{o:\top:\Prop}
\end{prooftree}
$$
where $o := \lambda A^{\Prop}. \lambda x^{A}. x$ and $\top := \forall A^{\Prop}.A \imp A$.
\end{lemma}
\begin{proof}
$$
\begin{array}{l@{\,}l@{~}l}
1.& \wfr{\Gamma}                                                       & \text{(hyp)} \\
2.& \dseqr{\Prop:\Type}                                                & \text{(ax 1)} \\
3.& \wfr{\Gamma,A:\Prop}                                               & \text{(env\textsubscript{2} 2)} \\
4.& \dseqr[\Gamma,A:\Prop]{A:\Prop}                                    & \text{(var 3)} \\
5.& \wfr{\Gamma,A:\Prop,x:A}                                           & \text{(env\textsubscript{2} 4)} \\
6.& \dseqr[\Gamma,A:\Prop,x:A]{x:A:\Prop}                              & \text{(var 5)} \\
7.& \dseqr[\Gamma,A:\Prop]{\lambda x^{A}.x:A \imp A:\Prop}             & \text{(abs+prod 6)} \\
8.& \dseqr{\lambda A^{\Prop}.\lambda x^{A}.x:\forall A^{\Prop}.A \imp A:\Prop}  & \text{(abs+prod 7)} \\
\end{array}
$$
\end{proof}
\section{\CCr meets the Poincar\'e criterion}\label{sec_ccr_pedag}
In this section we show that every type (term of sort $\Prop$ or $\Type$) in a well-formed environment of \CCr is inhabited. A sketch of the proof is: we first notice that in \CCr every product is inhabited, then, because each closed type reduces to a product, we can inhabit every type of a well-formed environment (beginning by its leftmost type, which is closed).

\begin{lemma}\label{cocr_prod_env}
If $\dseqr{\forall x^A.B : T}$ holds, then there exists $\kappa$ and a term $t$ such that $\dseqr{t:\forall x^A.B}$ and $T\beq \kappa$.
\end{lemma}
\begin{proof}
By induction on the derivation: if the last used rule is (prod) then we build $t$ by (abs) rule, and if it is (conv) then we apply induction hypothesis to get $t$.
\qed
\end{proof}

\begin{lemma}\label{cocr_lem_prod_inhab_type}
If $\dseqr{B:\Type}$ holds, then there exists a term $t$ such that $\dseqr{t:B}$ is derivable.
\end{lemma}
\begin{proof}
By cases on the last applied rule; (ax) case is dealt with lemma~\ref{cocr_not_empty}; (var), (app) and (conv) cases are eliminated using lemmas~\ref{lem_type_nowhere}  and~\ref{lem_type_of_type}; (prod) case is trivial using (abs) rule.
\qed
\end{proof}
Indeed, every element of type $\Type$ is syntactically of the form $\forall \vec{x}^{\vec{A}}.\Prop$, and then trivially inhabited by $\lambda \vec{x}^{\vec{A}}.\top$.

\begin{lemma}\label{cocr_lem_prod_inhab_prop}
If $\dseqr{B:\forall \vec{x}^{\vec{A}}.\Prop}$ holds with $B$ closed, then for all closed terms $w_1,\ldots,w_n$ verifying
$$
\begin{array}{c}
\dseqr{w_1:A_1} \\
\dseqr{w_2:A_2\dsub{x_1}{w_1}} \\
\vdots \\
\dseqr{w_n:A_n\dsub{x_1,\ldots,x_{n-1}}{w_1,\ldots,w_{n-1}}}
\end{array}
$$
there exists a term $t$ such that
$$
\dseqr{t:B\ \vec{w}}
$$
\end{lemma}
\begin{proof}
Let us define by $\norm{t}$ the length of the longest path of reduction from the term $t$ to its normal form (which exists because terms of \CCr are strongly normalizing).

We proceed by induction on the lexicographical order of $\norm{B\ \vec{w}}$ and the height of the derivation of $\dseqr{B:\forall \vec{x}^{\vec{A}}.\Prop}$.

Let us deal with non-trivial cases (others being mostly eliminated by lemmas~\ref{lem_type_nowhere} and~\ref{lem_type_of_type}):

\noindent(abs) If the last rule of the derivation is
	\dlineone{
		\drule{
		\dseqr[\Gamma,x_1:A_1]{u:\forall x_2^{A_2}\ldots\forall x_n^{A_n}.\Prop:\Type}
		}{
		\dseqr{\lambda x_1^{A_1}.u:\forall \vec{x}^{\vec{A}}.\Prop}
		}
	}
Let $\vec{w}$ be the above closed terms.

Substituting $v$ for $x_1$ in the premise, we obtain (property~\ref{prop_subst})
$$
\dseqr{u\dsub{x_1}{w_1}:\forall x_2^{A_2\dsub{x_1}{w_1}}\ldots\forall x_n^{A_n\dsub{x_1}{w_1}}.\Prop}
$$

As $\norm{u\dsub{x_1}{w_1}\ w_2\ ..\ w_n} < \norm{(\lambda x_1^{A_1}.u)\ w_1\ w_2\ ..\ w_n}$, and $u\dsub{x_1}{w_1}$ is closed (since $\lambda x_1^{A_1}.u$ and $w_1$ are), we can apply induction hypothesis to built a term $t$ such that $\dseqr{t:u\dsub{x_1}{w_1}\ w_2\ ..\ w_n}$ from which by (conv) rule we finally get
$$
\dseqr{t:(\lambda x_1^{A_1}.u)\ \vec{w}}
$$

\noindent(app) If the last rule of the derivation looks like
	\dlineone{
		\drule{
		\dseqr{u:\forall y^C.\forall \vec{x}^{\vec{D}}.\Prop} \quad \dseqr{v:C}
		}{
		\dseqr{u\ v:\forall\vec{x}^{\vec{D}\dsub{y}{v}}.\Prop}
		}
	}
	where $\vec{A} \equiv \vec{D}\dsub{y}{v}$ and $B \equiv u\ v$.

	Let $\vec{w}$ be the above terms. Since for every $i$ $x_i \not\in \VL(v)$, so
	$$
	D_i\dsub{y}{v}\dsub{x_1,..,x_{i-1}}{w_1,..w_{i-1}} \equiv D_i\dsub{y,x_1,..,x_{i-1}}{v,w_1,..w_{i-1}}
	$$
	Noticing we have $\norm{u\ v\ \vec{w}} = \norm{(u\ v)\ \vec{w}}$, we can then apply induction hypothesis of the first premise on the terms $v,\vec{w}$ to obtain $t$ such that
	$$
	\dseqr{t:(u\ v)\ \vec{w}}
	$$

\noindent(conv)
	\dlineone{
		\drule{
		\dseqr{B:T} \quad \dseqr{\forall \vec{x}^{\vec{A}}.\Prop:\Type} \quad T \beq \forall \vec{x}^{\vec{A}}.\Prop
		}{
		\dseqr{B:\forall \vec{x}^{\vec{A}}.\Prop}
		}
	}
	By lemma~\ref{lem_type_of_type} on $\dseqr{B:T}$, we have three cases: $T \equiv \Type$, $\dseqr{T:\Prop}$ or $\dseqr{T:\Type}$. By confluency, the definition of beta-reduction, the properties of subject reduction and uniqueness of types, only $\dseqr{T:\Type}$ remains. Hence $T$ must be of the form $\forall \vec{x}^{\vec{C}}.\Prop$ where $\vec{A} \beq \vec{C}$.
	
	Let $\vec{w}$ be the above terms. In order to apply induction hypothesis on the first premise, it is necessary to show that
	$$
	\begin{array}{c}
	\dseqr{w_1:C_1} \\
	\dseqr{w_2:C_2\dsub{x_1}{w_1}} \\
	\vdots \\
	\dseqr{w_n:C_n\dsub{x_1,\ldots,x_{n-1}}{w_1,\ldots,w_{n-1}}}
	\end{array}
	$$
	First let us notice that since $\vec{A} \beq \vec{C}$, then for each $i$ $A_i\dsub{x_1,..,x_{i-1}}{v_1,..,v_{i-1}}$ is convertible with $C_i\dsub{x_1,..,x_{i-1}}{v_1,..,v_{i-1}}$. Also, because $\dseqr{\forall \vec{x}^{\vec{C}}.\Prop:\Type}$, for each $i$ there exists $\kappa$ such that $\dseqr[\Gamma,x_1:C_1,..,x_i:C_i]{C_{i+1}:\kappa}$. \\
	We can then proceed by induction on $n$:
	$$
	\begin{array}{l@{\,}l@{\hspace{1em}}l}
	 1.& \dseqr{w_1:A_1}                                              & \text{(hyp)} \\
	 2.& \dseqr{C_1:\kappa}  \\
	 3.& A_1 \beq C_1  \\
	 4.& \mathbf{\dseqr{w_1:C_1}}                                     & \text{(conv 1 2 3)} \\

	 5.& \dseqr{w_2:A_2\dsub{x_1}{w_1}}                               & \text{(hyp)} \\
	 6.& \dseqr[\Gamma,x_1:C_1]{C_2:\kappa} \\
	 7.& \dseqr{C_2\dsub{x_1}{w_1}:\kappa}                            & \text{(prop.\ref{prop_subst} 4 6)} \\
	 8.& A_2\dsub{x_1}{w_1} \beq C_2\dsub{x_1}{w_1} \\
	 9.& \mathbf{\dseqr{w_2:C_2\dsub{x_1}{w_1}}}                      & \text{(conv 5 7 8)} \\

	&\vdots
	\end{array}
	$$

	Finally, we apply induction hypothesis of the first premise on those now well-typed $\vec{w}$ to get a term $t$ satisfying
	$$
	\dseqr{t:B\ \vec{w}}
	$$
\qed
\end{proof}

The two previous lemmas can be summed up by the following statement:

\begin{corollary}\label{cocr_cor_prod_inhab}
If $\dseqr{B:\kappa}$ holds with $B$ closed, then there exists a term $t$ such that $\dseqr{t:B}$.
\end{corollary}

So the pedagogical character of the calculus follows, every type of a well-formed environment is inhabited:

\begin{theorem}[(Poincar\'e criterion)]\label{cocr_pedag}
If $\wfr{x_1:A_1, \ldots, x_n:A_n}$ holds, then there exists terms $t_1, \ldots, t_n$ such that
$$
\begin{array}{c}
\dseqr[]{t_1:A_1} \\
\dseqr[]{t_2:A_2\dsub{x_1}{t_1}} \\
\vdots \\
\dseqr[]{t_n:A_n\dsub{x_1,\ldots,x_{n-1}}{t_1,\ldots,t_{n-1}}}
\end{array}
$$
\end{theorem}
\begin{proof} By induction on the size of the environment $n$.

From the derivation $\wfr{x_1:A_1, \ldots, x_n:A_n}$, we have $\dseqr[]{A_1:\kappa}$ as a sub-derivation where $A_1$ is closed. So by corollary~\ref{cocr_cor_prod_inhab}, we get $t_1$ such that
$$
\dseqr[]{t_1:A_1}
$$
Then by property~\ref{prop_subst} we have $\wfr{x_2:A_2\dsub{x_1}{t_1}, \ldots, x_n:A_n\dsub{x_1}{t_1}}$.
By the same way, we construct $t_2$ such that
$$
\dseqr[]{t_2:A_2\dsub{x_1}{t_1}}
$$
and then $\wfr{x_3:A_3\dsub{x_1,x_2}{t_1,t_2}, \ldots, x_n:A_n\dsub{x_1,x_2}{t_1,t_2}}$.
$$
\vdots
$$
\qed
\end{proof}

This so named ``motivation'' may be transmitted to the conclusion of judgments:

\begin{corollary}\label{cocr_cor_pedag}
If $\dseqr[x_1:A_1, \ldots, x_n:A_n]{u:B}$ holds, then there exists terms $t_1, \ldots, t_n$ such that
$$
\begin{array}{c}
\dseqr[]{t_1:A_1} \\
\dseqr[]{t_2:A_2\dsub{x_1}{t_1}} \\
\vdots \\
\dseqr[]{t_n:A_n\dsub{x_1,\ldots,x_{n-1}}{t_1,\ldots,t_{n-1}}}
\end{array}
$$
and
$$
\dseqr[]{u\dsub{\vec{x}}{\vec{t}} : B\dsub{\vec{x}}{\vec{t}} }
$$
\end{corollary}
\begin{proof} Immediate by applying $n$ times the property~\ref{prop_subst} using the terms obtained from the theorem.
\qed
\end{proof}

\begin{theorem}[(usefulness)]\label{cocr_usefulness}
If $\dseqr[]{f:\forall x^A.B}$ holds, then there exists a term $u$ such that $\dseqr[]{u:A}$.
\end{theorem}
\begin{proof}
From $\dseqr[]{f:\forall x^A.B}$, by lemma~\ref{lem_type_of_type} we have $\dseqr[]{\forall x^A.B:\kappa}$, then $\dseqr[x:A]{B:\kappa}$ which implies that $\wf{x:A}$, and finally by theorem~\ref{cocr_pedag} we construct $u$.
\qed
\end{proof}

\section{Limitations of the logical power of \CCr}\label{sec_limitations}
To introduce an hypothesis (which is not a variable) in an environment, it is necessary to first inhabit it. For instance, defining Leibniz equality over a type $A$ by
$$
x =_A y := \forall Q^{A \imp \Prop}. Q\ x \imp Q\ y
$$
it is not possible to prove nor symmetry nor transitivity of this relation over $A$ (whatever this type is). Indeed, because we are not permitted to derive $\dseqr[A:\Prop,x:A,y:A]{x =_A y:\Prop}$, we can not introduce $x =_A y$ as an hypothesis and then we are not allowed to use it.

\begin{theorem}
There is no term $u$ such that $\dseqr[]{u:\forall A^{\Prop}.\forall x^A.\forall y^A.x =_A y \imp y =_A x}$ holds.
\end{theorem}
\begin{proof}
Let us suppose such a term $u$ exists. So we have a sort $\kappa$ such that $\dseqr[A:\Prop,x:A,y:A]{x =_A y : \kappa}$. And because $x =_A y$ is a product, by lemma~\ref{cocr_prod_env}, it is inhabited, say by $t$. But since \CCr is a restriction of \CC, $\dseq[A:\Prop,x:A,y:A]{t:x =_A y}$ also holds in \CC. Then, applying it to $\mathbb{N}$ and $0$ and $1$, we get a proof of $0=1$ in the empty environment in \CC, which is known to be impossible (by a simple combinatoric discussion about the normal form of such a term).
\qed
\end{proof}

In fact, this calculus does not even natively contain simply typed $\lambda$-calculus:
\begin{theorem}
There is no term $u$ such that
$$
\dseqr[A\ B\ C:\Prop]{u: (A \imp B) \imp (B \imp C) \imp (A \imp C)}
$$
holds.
\end{theorem}
\begin{proof}
Using same arguments as above, if such a $u$ exists, then the following judgment holds:
$$
\dseqr[A:\Prop, B:\Prop, C:\Prop]{A \imp B:\Prop}
$$
so there is an inhabitant $t$ of the product type $A \imp B$ in \CCr and hence in \CC, implying by (abs) rule that
$$
\dseq[]{\lambda A B C^{\Prop}.t:\forall A B C^{\Prop}.A \imp B}
$$
which can be specialized to $\top$ and $\bot$ to obtain a proof of $\top \imp \bot$ and finally a proof of $\bot$ in the empty environment, which is impossible since \CC is consistent.
\qed
\end{proof}

Actually, every instances of the types in \CCr must be inhabited:
\begin{theorem}
If $\dseqr[x_1:A_1,..,x_n:A_n]{B:\kappa}$ holds, then for all terms $w_1,\ldots,w_n$ such that
$$
\begin{array}{c}
\dseqr[]{w_1:A_1} \\
\dseqr[]{w_2:A_2\dsub{x_1}{w_1}} \\
\vdots \\
\dseqr[]{w_n:A_n\dsub{x_1,\ldots,x_{n-1}}{w_1,\ldots,w_{n-1}}}
\end{array}
$$
there exists a term $t$ such that
$$
\dseqr[]{t:B\dsub{\vec{x}}{\vec{w}}}
$$
\end{theorem}
\begin{proof}
The proof is trivial by applying $n$ times the substitution property~\ref{prop_subst}, obtaining $\dseqr[]{B\dsub{\vec{x}}{\vec{w}}:\kappa}$, inhabited by corollary~\ref{cocr_cor_prod_inhab}.
\qed
\end{proof}

It is hard to precisely determine the logical expressiveness of \CCr. We have at least simply typed $\lambda$-calculus on closed (and then inhabited) types of \CCr (e.g. $\top$, $\mathbb{N}$, etc.). The proof is the same as the one of lemma~\ref{simple_types_inhab} below.
\section{Computational expressivity of \CCr}\label{sec_expressivness}
Although the logical strength of \CCr seems quite poor, its computational power is at least that of the G\"odel system T. We use the usual well-known way to define terms, types (except cartesian product), and recursor (from iterator) of system T in lambda-calculus (see\cite{JGPT90}).

\begin{definition}
$$
\begin{array}{r@{\ :=\ }l}

\mathbb{N}           & \forall A^{\Prop}. A \imp (A \imp A) \imp A \\
0                    & \lambda A^{\Prop}.\lambda x^A.\lambda f^{A \imp A}.x \\
S(n)                 & \lambda A^{\Prop}.\lambda x^A.\lambda f^{A \imp A}.f\ (n\ A\ x\ f) \\
it_T(n,b,(y^T)step) & n\ T\ b\ (\lambda y^T.step)

\end{array}
$$
\end{definition}

\begin{lemma}
The following rules are derivable:
\dlinetwo{
	\drule{
		\wfr{\Gamma}
	}{
		\dseqr{0:\mathbb{N}:\Prop}
	}
}{
	\drule{
		\dseqr{n:\mathbb{N}}
	}{
		\dseqr{S(n):\mathbb{N}}
	}
}
\dlineone{
	\drule{
		\dseqr{T:\Prop} \quad  \dseqr{n:\mathbb{N}} \quad \dseqr{b:T} \quad \dseqr[\Gamma,y:T]{step:T}
	}{
		\dseqr{it_T(n,b,(y^T)step):T}
	}
}
\end{lemma}

\begin{lemma}
The following reductions hold:
$$
\begin{array}{r@{\ \breds\ }l}
it_T(0,b,(y^T)step) & b \\
it_T(S(n),b,(y^T)step) & step\dsub{y}{it_T(n,b,(y^T)step)}
\end{array}
$$
\end{lemma}

\begin{definition}[(simple types on $\mathbb{N}$)]\label{def_simple_types}
Simple types on $\mathbb{N}$ are those obtained from $\mathbb{N}$ and $\imp$.
\end{definition}

\begin{lemma}\label{simple_types_inhab}
If $\wfr{\Gamma}$ holds and $T$ is a simple type on $\mathbb{N}$, then there exists a term $t$ such that $\dseqr{t:T:\Prop}$.
\end{lemma}
\begin{proof}
By induction on $T$ (as a simple type on $\mathbb{N}$):
\begin{itemize}
\item If $T$ is $\mathbb{N}$, then $0$ fits.
\item If $T$ is $A \imp B$ where $A$ and $B$ are simple types on $\mathbb{N}$, then by induction hypothesis on $A$, we get $\dseqr{A:\Prop}$ and by (env\textsubscript{2}) rule we obtain $\wfr{\Gamma,x:A}$. By induction hypothesis on $B$, we get $\dseqr{b:B:\Prop}$, and weakening it we have $\dseqr[\Gamma,x:A]{b:B:\Prop}$, and finally, by (abs) and (prod) rules, $\dseqr{\lambda x^A.b:A \imp B:\Prop}$.
\end{itemize}
\qed
\end{proof}

\bigskip

\CCr does not allow us to derive the usual cartesian product defined by $A \times B := \forall C^{\Prop}.(A \imp B \imp C) \imp C$. To simulate recursor from iterator, we define a restricted cartesian product $\mathbb{N} \times T$ for each $T$, simple type on $\mathbb{N}$, by encoding a natural into $T$.

\begin{lemma}
If $\wfr{\Gamma}$ holds and $T$ is a simple type on $\mathbb{N}$ then there exists two terms $enc_T$ and $dec_T$ such that $\dseqr{enc_T:\mathbb{N} \imp T}$ and $\dseqr{dec_T:T \imp \mathbb{N}}$ and for every term $n$ we have $dec_T(enc_T\ n) \breds n$.
\end{lemma}
\begin{proof}
By induction on $T$ (as a simple type on $\mathbb{N}$):
\begin{itemize}
\item If $T$ is $\mathbb{N}$, then we take the identity on $\mathbb{N}$ for $enc_T$ and $dec_T$.
\item If $T$ is $A \imp B$, we take
$$
\begin{array}{r@{\ :=\ }l}
enc_{A \imp B} & \lambda x^{\mathbb{N}}.\lambda z^A.enc_B\ x \\
dec_{A \imp B} & \lambda f^{A \imp B}.dec_B\ (f\ a)
\end{array}
$$
where $a$ is a term of type $A$ obtained from lemma~\ref{simple_types_inhab}.
\end{itemize}
\qed
\end{proof}

\begin{definition}
We define the following abbreviations for couples
$$
\begin{array}{r@{\ :=\ }l}
\mathbb{N} \times T & (T \imp T \imp T) \imp T \\
\cpl{n}{t}^T & \lambda f^{T \imp T \imp T}. f\ (enc_T\ n)\ t \\
\pi_1(c) & dec_T\ (c\ (\lambda x^T.\lambda y^T.x)) \\
\pi_2(c) & c\ (\lambda x^T.\lambda y^T.y)
\end{array}
$$
\end{definition}

\begin{lemma}
The following rules are derivable:
\dlinetwo{
	\drule{
		\wfr{\Gamma}
	}{
		\dseqr{\mathbb{N} \times T:\Prop}
	}
}{
	\drule{
		\dseqr{n:\mathbb{N}} \quad \dseqr{t:T}
	}{
		\dseqr{\cpl{n}{t}^T:\mathbb{N} \times T}
	}
}
\dlinetwo{
	\drule{
		\dseqr{c:\mathbb{N} \times T}
	}{
		\dseqr{\pi_1(c):\mathbb{N}}
	}
}{
	\drule{
		\dseqr{c:\mathbb{N} \times T}
	}{
		\dseqr{\pi_2(c):T}
	}
}
\end{lemma}

\begin{lemma}
The following reductions hold:
$$
\begin{array}{r@{\ \breds\ }l}
\pi_1(\cpl{n}{t}^T) & n \\
\pi_2(\cpl{n}{t}^T) & t
\end{array}
$$
\end{lemma}

\begin{definition}
We define recursor from iterator by
$$
rec_T(n,b,(x^{\mathbb{N}},y^T)step) := \pi_2 \left[ it_{T \times T}(n,\cpl{0}{b}^T,(z^{T\times T})step') \right]
$$
where
$$
step' := \cpl{S(\pi_1(z)) }{ step\dsub{x,y}{\pi_1(z),\pi_2(z)}}^{T \times T}
$$
\end{definition}

\begin{lemma}
The following rule is derivable:
\dlineone{
	\drule{
		\dseqr{T:\Prop} \quad  \dseqr{n:\mathbb{N}}  \quad  \dseqr{b:T}  \quad  \dseqr[\Gamma,x:\mathbb{N},y:T]{step:T}
	}{
		\dseqr{rec_T(n,b,(x^{\mathbb{N}},y^T)step):T}
	}
}
\end{lemma}

\begin{lemma}
The following reductions hold:
$$
\begin{array}{r@{\ }l}
rec_T(0,b,(x^{\mathbb{N}},y^T)step)    & \breds b \\
rec_T(S(n),b,(x^{\mathbb{N}},y^T)step) & \breds step\dsub{x,y}{n,rec_T(n,b,(x^{\mathbb{N}},y^T)step)}
\end{array}
$$
\end{lemma}

\section{Conclusions and direction for further work}

We have seen a simple attempt to pedagogize the calculus of constructions. It has a good computational power ---at least G\"odel system T--- but lacks of logical expressivity ---does not even natively contain simply typed $\lambda$-calculus. A pleasant aspect is the simplicity of the added constraint, which also emphasizes that the (prod) rule is responsible for vacuity in \CC.

\bigskip

Logical limitations of our calculus \CCr suggest a more precise definition for a calculus of constructions to be pedagogical: in a pedagogical calculus, we should be able to prove the symmetry of the Leibniz equality, because the non-emptiness of $x =_A y$ can be justified by substituting $\mathbb{N}$ to $A$ and $0$ to $x$ and $y$. It means that we not only need that a well-formed environment guarantees the non-emptiness of its types by exhibiting an example, but the converse should hold too.

\medskip

But as it was already pointed out in section~\ref{sec_naive_extension}, the direct converse statement of the Poincar\'e criterion is not suitable. We then propose the following definition of a pedagogical subsystem of \CC (whose judgments are indexed by p):
\begin{definition}[(pedagogical subsystem of \CC)]\hspace*{1em}\\
$P$ is a pedagogical subsystem of \CC if:
\begin{enumerate}
\item
$\wfp{x_1:A_1, \ldots, x_n:A_n}$ holds {\bf if and only if}
	\begin{enumerate}
	 \item $\wf{x_1:A_1, \ldots, x_n:A_n}$ holds in \CC,
	 \item and there exist terms $t_1, \ldots, t_n$ such that
	 $$
	 \begin{array}{c}
	 \dseqp[]{t_1:A_1:\kappa_1} \\
	 \dseqp[]{t_2:A_2\dsub{x_1}{t_1}:\kappa_2} \\
	 \vdots \\
	 \dseqp[]{t_n:A_n\dsub{x_1,\ldots,x_{n-1}}{t_1,\ldots,t_{n-1}}:\kappa_n}
	 \end{array}
	 $$
	 \end{enumerate}
\item the system is stable by reduction, namely if $\dseqp{u:B}$ and $u \bred u'$, then $\dseqp{u':B}$.
\end{enumerate}
\end{definition}

\begin{remark}
\begin{enumerate}
 \item The left to right side of the equivalence is already known as ``the Poincar\'e criterion'', and enforces $P$ to be a subsystem of \CC. The right to left side should then be named ``the converse of the Poincar\'e criterion''.
 \item The subject reduction must be explicitly stated here since\cite{LCDM08} defined a ``simple pedagogical second-order propositional calculus (P\textsubscript{s}-Prop\textsubscript{2})'' verifying 1 but not 2.
\end{enumerate}
\end{remark}

\medskip

One can show, keeping only the rules of \CC necessary to define second order $\lambda$-calculus and adding constraints of the pedagogical second order $\lambda$-calculus of~\cite{LCDM09}, that we obtain a calculus which is pedagogical in the new sense just defined. For instance, P-MPC et P-Prop$^2$ [see section~\ref{sec_intro}] satisfy: it exists $F$ such that $\Gamma \vdash F$ {\em if and only if} it exists $\sigma$ such that $\vdash \sigma \cdot \Gamma$.

By the same way, we can construct more expressive pedagogical restrictions of \CC: a hint is given by\cite{DM08} where he studies pedagogical propositional higher order systems. It thus raises the question of formally characterizing a maximally expressive pedagogical restriction of \CC.

\bibliographystyle{jucs}

\nocite{PG53,GG46,GG50,GG51,GG51-a,VNK00a,VNK00b,VM75,DN66,DN73,VV55,PV53}
\bibliography{bibliography}

\end{document}